\documentclass[letterpaper, 10 pt, conference]{ieeeconf}  % Comment this line out if you need a4paper

\IEEEoverridecommandlockouts                              % This command is only needed if 
\usepackage{amsmath} % assumes amsmath package installed
\usepackage{amssymb}  % assumes amsmath package installed
\usepackage{commath}
\usepackage{soul,color}
\usepackage{graphicx}
\usepackage{mathptmx}
\usepackage{flushend}
\newcommand{\tr}{^\mathrm{T}}
\def\define{\buildrel \triangle \over =}
\newtheorem{theorem}{Theorem}

\usepackage{color}
\newcommand{\ri}{\textcolor{red}}

\usepackage{lscape}
\allowdisplaybreaks

\title{\LARGE \bf
State of Charge Estimation of Parallel Connected Battery Cells \\ via Descriptor System Theory*
}

\author{Dong Zhang$^{1}$, Luis D. Couto$^{2}$, Sebastien Benjamin$^{3}$, Wente Zeng$^{4}$, Daniel F. Coutinho$^{5}$ and Scott J. Moura$^{1}$% <-this % stops a space
\thanks{*This work is sponsored by Total S.A. and Saft Batteries.}% <-this % stops a space
\thanks{$^{1}$Dong Zhang and Scott J. Moura is with Department of Civil and Environmental Engineering, University of California, Berkeley, California, 94720,
USA. {\tt\small \{dongzhr,smoura\}@berkeley.edu}}%
\thanks{$^{2}$Luis D. Couto is with Department of Control Engineering and System Analysis, Universit\'e Libre de Bruxelles, B-1050 Brussels, Belgium. {\tt\small lcoutome@ulb.ac.be}}
\thanks{$^{3}$Sebastien Benjamin is with Saft Batteries, Total S.A.. {\tt\small Sebastien.BENJAMIN@saftbatteries.com}}
\thanks{$^{4}$Wente Zeng is with Total S.A.. {\tt\small wente.zeng@total.com}}
\thanks{$^{5}$Daniel F. Coutinho is with Department of Automation and Systems, UFSC, PO BOX 476, 88040-900, Florian{\'o}polis, SC, Brazil. {\tt\small daniel.coutinho@ufsc.br}}%
}

\begin{document}

\maketitle
\thispagestyle{empty}
\pagestyle{empty}

%%%%%%%%%%%%%%%%%%%%%%%%%%%%%%%%%%%%%%%%%%%%%%%%%%%%%%%%%%%%%%%%%%%%%%%%%%%%%%%%
\begin{abstract}

This manuscript presents an algorithm for individual Lithium-ion (Li-ion) battery cell state of charge (SOC) estimation when multiple cells are connected in parallel, using only terminal voltage and total current measurements. For battery packs consisting of thousands of cells, it is desirable to estimate individual SOCs by only monitoring the total current in order to reduce sensing cost. Mathematically, series connected cells yield dynamics given by ordinary differential equations under classical full voltage sensing. In contrast, parallel connected cells are evidently more challenging because the dynamics are governed by a nonlinear descriptor system, including differential equations and algebraic equations arising from voltage and current balance across cells. An observer with linear output error injection is formulated, where the individual cell SOCs and local currents are locally observable from the total current and voltage measurements.
The asymptotic convergence of differential and algebraic states is established by considering local Lipschitz continuity property of system nonlinearities. 
Simulation results on LiNiMnCoO$_2$/Graphite (NMC) cells illustrate convergence for SOCs, local currents, and terminal voltage. 

\end{abstract}

%%%%%%%%%%%%%%%%%%%%%%%%%%%%%%%%%%%%%%%%%%%%%%
\section{Introduction}
%%%%%%%%%%%%%%%%%%%%%%%%%%%%%%%%%%%%%%%%%%%%%%

Lithium-ion (Li-ion) batteries have emerged as one of the most prominent energy storage devices for large-scale energy applications, e.g., hybrid electric vehicle (HEV), battery electric vehicles (BEV) and smart grids, due to their high energy and power density, low self-discharge and long lifetime \cite{chaturvedi2010algorithms}. A battery pack system generally consists of hundreds or thousands of single cells connected in parallel and series connections in order to fulfill the requirements of high-energy and high-power applications \cite{Zhong-2014}. It is well-known that Li-ion cells are sensitive to overcharge and over-discharge \cite{Lin-2015}. An accurate estimation of the internal states, including state of charge (SOC), enables a battery management system (BMS) to prolong battery service life by ensuring individual cells within a pack do not overcharge or over-discharge.

Battery pack system modeling can be divided into three categories. The first approach treats the entire pack as one lumped single cell \cite{Castano-2015}. However, the internal states of individual cells within the pack are likely to be different, due to parameter heterogeneity. Therefore, some cells are more prone to violate safety-critical constraints than others, which cannot be resolved from the lumped single cell approach. The second modelling approach also relies on a single cell model, but it focuses on a set of specific in-pack cells -- the weakest and the strongest ones, as representatives of the pack dynamics \cite{Zhong-2014,Hua-2015}. The last modelling approach is based on the interconnection of single cell models \cite{Zheng-2013,Zhang-2016,Zhao-2015}. This approach benefits from high fidelity cell-by-cell resolution, but it might suffer from high real-time computational burden. To counteract this computational challenge, most of these approaches resort to equivalent circuit models, which tend to have a low complexity when compared to more sophisticated electrochemical models.

The state estimation problem for series arrangements of battery cells has been studied previously \cite{Lin-2015,hu2010estimation}, whereas the estimation for cells in parallel has been overlooked for multiple reasons. First, cells in parallel are widely considered to behave as one single cell. However, an implicit assumption behind this reasoning is that the applied current is evenly split amongst the cells in parallel. This is hardly true in practice due to cell heterogeneities, such as non-uniform parameter values and temperatures \cite{Bruen-2016}. 
This fact makes the estimation problem for parallel battery cells relevant. Secondly, the estimation problem for battery cells in series is arguably easier to solve than the parallel counterpart, because in the series case the input current to each battery cell is the same and it can be practically measured. In the parallel case, each cell's local currents are unknown and determined by algebraic constraints. Due to sensing limitations, only the total current can be measured. Therefore, the parallel configuration turns out to be a differential algebraic equation (DAE) system that requires non-trivial estimation theories.

A DAE system, \emph{a.k.a.} descriptor system, involving both differential and algebraic equations, is a powerful modelling framework that generalizes ordinary differential (normal) systems \cite{duan2010analysis}. 
The state observer design for linear descriptor systems is a rich research topic \cite{duan2010analysis,Niko-1992,Chisci-1992,Darouach-1995}. In contrast, state observers for nonlinear descriptor systems is less prolific. Some relevant contributions encompass a local asymptotic state observer %relying on linearization techniques 
\cite{Boutayeb-1995}, looking at the system as differential equations on a restricted manifold %and designing the observer through standard techniques for explicit systems 
\cite{Zimmer-1997}, and an index-1 DAE observer \cite{Aslund-2006}. 
%Lipschitz systems
Other works consider the case of Lipschitz nonlinearities \cite{rajamani1998observers}, which have served as a basis for Lyapunov-based observer design using the linearized system \cite{kaprielian1992observer}, and LMI approaches producing state observers in singular \cite{Lu-2006} and non-singular \cite{Darouach-2008} forms. Another Lipschitz system was considered in \cite{Shields-1997}, where the temporal separation between slow and fast dynamics was exploited to design a robust state observer. 
Nonlinear descriptor systems have also been estimated through moving horizon approaches \cite{Albuquerque-1997} and Kalman filters \cite{Becerra-2001,Mandela-2010,Puranik-2012}.

\iffalse
a powerful modelling framework that generalizes ordinary (normal) systems\cite{duan2010analysis}. Descriptor systems arise naturally from the physics in many fields, including power systems, electrical networks and interconnected large-scale systems, to name a few. These systems come associated with added complexity involving e.g. infinite frequencies, impulsive behaviour, regularity conditions (existence and uniqueness) and so on. Such complexities can be bypassed by transforming the descriptor system into a normal one (reduced model) through algebraic state elimination. However, this process might not always be possible or desirable. If the original descriptor system is nonlinear, the reduced model can exhibit a more complex structure than the original one. Besides, in the estimation context, a state observer for the reduced system cannot cope with uncertainty in the algebraic equations and the algebraic state is reconstructed in an open loop fashion. These different aspects serve as motivation to pursue a state observer for a cells-in-parallel battery pack modelled as a nonlinear descriptor system.
\fi

In light of the aforementioned literature, the contributions of this paper are threefold:
\begin{enumerate}
    \item Propose a novel framework for modelling Li-ion battery cells in parallel as a nonlinear descriptor system;
    \item Conduct observability analysis of such a system;
    \item Design a Lyapunov-based asymptotic state observer for both differential and algebraic state estimation, using only voltage and total current measurements.
\end{enumerate}

The reminder of this paper is organized as follows. Section \ref{s:mod} introduces the modelling framework for parallel cells. Section \ref{s:mot} motivates the importance of observer design with cell heterogeneity. Section \ref{s:obs_ana} provides the local observability analysis for the nonlinear descriptor system. Section \ref{s:obs_design} discusses the state observer design and its convergence analysis. Finally, the effectiveness of the proposed approach is illustrated in Section \ref{s:sim} via numerical simulations. Conclusions are drawn in Section \ref{s:conclusion}.

\textbf{Notation.} Throughout the manuscript, the symbols $I_{p \times q}$ and $\mathbf{0}_{p \times q}$ denote the identity matrix and the zero matrix with dimension $p \times q$, respectively. The inner product between $x,y\in\mathbb{R}^n$ is given by
\begin{equation}
    \langle x,y \rangle = \sum_{i = 1}^n x_iy_i. \nonumber
\end{equation}

%%%%%%%%%%%%%%%%%%%%%%%%%%%%%%%%%%%%%%%%%%%%%%
\section{Parallel Battery Model Formulation} \label{s:mod}
%%%%%%%%%%%%%%%%%%%%%%%%%%%%%%%%%%%%%%%%%%%%%%

This section first reviews an equivalent circuit model (ECM) for a single battery cell, which is then electrically interconnected with other cell models to form a parallel arrangement of cells.

\subsection{Single Battery Cell}

Consider the ECM for a single battery cell, shown in \mbox{Fig. \ref{ECM}}, represented by the following continuous-time state-space representation,
\begin{align}
    \label{e:ecm}
    \dot{x}_k(t) & = \overline{A}_{k} x_k(t) + \overline{B}_{k} u_k(t), \\
    y_k(t) & = h_{k}(x_k(t)) + \overline{D}_{k} u_k(t),     \label{e:oecm}
\end{align}
where $x_k \in \mathbb{R}^2$ is the state vector for $k$-th battery cell in the parallel connection defined as
\begin{equation}
x_k \define [z_k, \ V_{c,k}]^\top, \nonumber %\label{e:state}
\end{equation}
with $z_k$ as the SOC and $V_{c,k}$ as the capacitor voltage of the RC pair for the $k$-th battery cell. In \eqref{e:ecm}, $u_k \in \mathbb{R}$ is the applied current $u_k = I_k(t)$, and state matrix $\overline{A}_{k} {\in \mathbb{R}^{2\times2}}$ and input matrix $\overline{B}_{k}{\in \mathbb{R}^{2\times 1}}$ are given by
\begin{equation}
\overline{A}_{k} = \begin{bmatrix}0&0\\0&-\frac{1}{R_{2,k}C_k}\end{bmatrix}, \quad
\ \overline{B}_{k} = \begin{bmatrix}\frac{1}{Q_k} \\ \frac{1}{C_k}\end{bmatrix},
\label{e:ecmmat}
\end{equation}
where $Q_k$ represents the capacity of cell $k$, and $R_{1,k}$, $R_{2,k}$, $C_k$ are resistances and capacitance shown in Fig. \ref{ECM}. The output equation \eqref{e:oecm} for the $k$-th cell provides the voltage response characterized by the nonlinear function,
\begin{equation}
    \label{e:ohecm}
    h_{k}(x_k) = OCV(z_k) + V_{c,k},
\ \overline{D}_{k} = R_{1,k}\!,
\end{equation}
where $y_k \in \mathbb{R}$ is the battery terminal voltage, function $h_{k}:\mathbb{R}^2\rightarrow \mathbb{R}$ consists of the open circuit voltage as a function of SOC denoted as $OCV(z_k)$, voltage across the RC pair $V_{c,k}$, and voltage response due to an ohmic resistance $R_{1,k}$.
\begin{figure}[t]
	\centering
	\includegraphics[trim = 0mm 0mm 0mm 0mm, clip, width=0.8\linewidth]{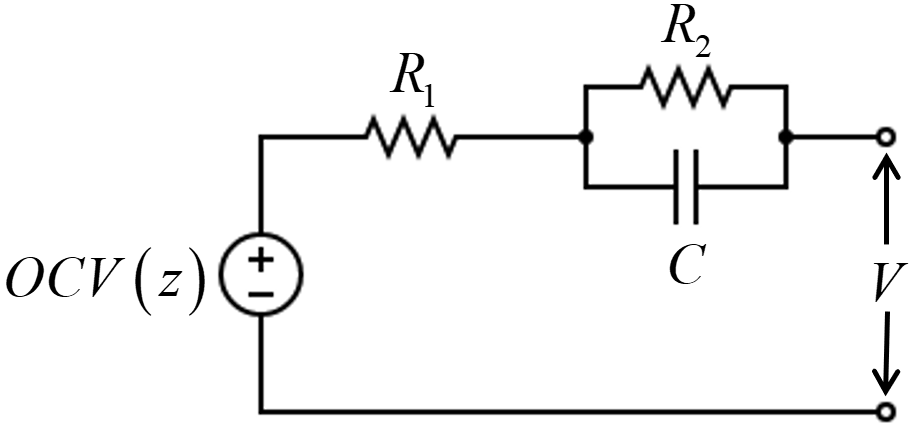}
	\caption{Schematic for an equivalent circuit model (ECM)}
	%\vspace{-4ex}
	\label{ECM}
\end{figure}

\subsection{Parallel Arrangement of Battery Cells}

For a block of $n$ cells in parallel, in order to reduce sensing effort, we assume only the total current and voltage for one of the cells are measured. Electrically, Kirchhoff's voltage law indicates that a parallel connection of cells constraints terminal voltage to the same value for all cells. Kirchhoff's current law indicates that the overall current is equal to the summation of cell local currents. Mathematically, the following nonlinear algebraic constraints, according to Kirchhoff's voltage law, need to be enforced:
\begin{align}
    \label{volt_constraint}
    OCV(z_i) + V_{c,i} + R_{1,i} I_i = & OCV(z_j) + V_{c,j} + R_{1,j} I_j, \nonumber \\
    & \forall i,j\in\{1,2,\cdots,n\}, \; i\neq j.
\end{align}
Similarly, Kirchhoff's current law poses the following linear algebraic constraint with respect to cell local currents,
\begin{align}
    \label{cur_constraint}
    \sum_{k = 1}^{n} I_k(t) = I(t),
\end{align}
where $I(t)$ is the total current applied to the parallel battery system. It is worth highlighting that (\ref{volt_constraint}) imposes $(n-1)$ nonlinear algebraic constraints with respect to differential states and local currents, whereas (\ref{cur_constraint}) imposes 1 algebraic constraint with respect to local currents. In this manuscript, it is assumed that all cells have different electrical model parameters. In addition, when only the total current is measured, the local currents of cells are unknown. Hence, the system of differential-algebraic equations must be solved such that the algebraic equations (\ref{volt_constraint}) and (\ref{cur_constraint}) are fulfilled for all $t$. Such methodology is realized by augmenting the local currents to the differential state vector to form a nonlinear descriptor system \cite{duan2010analysis}, which takes the form
\begin{align}
    \label{descriptor_dyn}
    E \dot{w}(t) = & A w(t) + \theta(w(t)), \\
    \label{descriptor_out}
    y(t) = & Hw(t) + \phi(w(t)),
\end{align}
where $w = [x\quad u]^\top \in \mathbb{R}^{3n}$ with
\begin{align}
    x & = \begin{bmatrix}x_1 & x_2 & \cdots & x_{n}\end{bmatrix}^\top \in \mathbb{R}^{2n}, \\
    u & = \begin{bmatrix}I_1 & I_2 & \cdots & I_n\end{bmatrix}^\top \in \mathbb{R}^{n}, \\
    y & = \begin{bmatrix}y_1 & y_2 & \cdots & y_n\end{bmatrix}^\top \in \mathbb{R}^{n}.
\end{align}
Equation (\ref{descriptor_dyn}) encodes both the system dynamical equations and algebraic equations, and the matrix $E$ is a singular matrix of the form
\begin{equation}
    \label{E_mat}
    E = \begin{bmatrix}I_{2n \times 2n} & \mathbf{0}_{2n \times n} \\ \mathbf{0}_{n \times 2n} & \mathbf{0}_{n \times n}\end{bmatrix} \in \mathbb{R}^{3n \times 3n},
\end{equation}
where $I_{2n \times 2n}$ is an identity matrix of size $2n$-by-$2n$. Matrix $A$ accounts for the linear part of the system equations with
\begin{equation}
    A = \begin{bmatrix} A_{11} & A_{12} \\ A_{21} & A_{22} \end{bmatrix} \in \mathbb{R}^{3n \times 3n}, \label{e:Amatrix}
\end{equation}
where
\begin{align}
    \hspace{-0.15cm} A_{11} & = \mathrm{diag}(\overline{A}_{1},\overline{A}_{2},\cdots,\overline{A}_{n}), \ A_{12} = \mathrm{diag}(\overline{B}_{1},\overline{B}_{2},\cdots,\overline{B}_{n}), \\
    \hspace{-0.15cm} A_{21} & = \begin{bmatrix}
    0 & 1 & S & 0 & \cdots & 0 \\
    0 & 1 & 0 & S & \cdots & 0 \\
    \vdots & \vdots & \vdots & \vdots & \ddots & \vdots \\
    0 & 1 & 0 & 0 & \cdots & S \\
    0 & 0 & 0 & 0 & \cdots & 0 \\
    \end{bmatrix} \in \mathbb{R}^{n \times 2n}, \quad
    S = \begin{bmatrix} 0 & -1 \end{bmatrix}, \nonumber \\
    \hspace{-0.15cm} A_{22} & = \begin{bmatrix}
    R_{1,1} & -R_{1,2} & 0 & \cdots & 0 \\
    R_{1,1} & 0 & -R_{1,3} & \cdots & 0 \\
    \vdots & \vdots & \vdots & \ddots & \vdots \\
    R_{1,1} & 0 & 0 & \cdots & -R_{1,n} \\
    1 & 1 & 1 & \cdots & 1 \\
    \end{bmatrix} \in \mathbb{R}^{n \times n}.
    \label{e:A22mat}
\end{align}
Notice that matrix $A_{22}$ is full rank, i.e. the linear part of the descriptor model is regular and impulsive free.

Function $\theta(w)$ constitutes the nonlinear portion in the system equations from the voltage algebraic constraints (\ref{volt_constraint}):
\begin{equation}
    \theta(w) = \begin{bmatrix}
    \theta_x(w) \\
    \theta_u(w)
    \end{bmatrix} = 
    \begin{bmatrix}
    \mathbf{0}_{2n \times 1} \\
    OCV(z_1) - OCV(z_2) \\
    \vdots \\
    OCV(z_1) - OCV(z_n) \\
    -I(t)
    \end{bmatrix} \in \mathbb{R}^{3n},
\end{equation}
where $\theta_x \in \mathbb{R}^{2n}$ represents nonlinearity in the dynamical equations and corresponds to row 1 through row $2n$ of $\theta(w)$, and $\theta_u \in \mathbb{R}^{n}$ is nonlinearity appears in the algebraic equations and corresponds to row $(2n+1)$ through row $3n$ of $\theta(w)$. The output (\ref{descriptor_out}) models the voltage of each battery cell, with
\begin{equation}
    H = \left[H_x \;\; H_u\right], \quad \phi(w) = \left[ OCV(z_1) \;\; \cdots \;\; OCV(z_n) \right]^\top
    \label{e:Hmatrix}
\end{equation}
where $H_x$ corresponds to column 1 through column $2n$ of matrix $H$, and $H_u$ corresponds to column $(2n+1)$ through column $3n$ of matrix $H$:
\begin{align}
    H_x & = \mathrm{diag}(-S,-S,\cdots,-S) \in \mathbb{R}^{n\times2n}, \\
    H_u & = \mathrm{diag}(R_{1,1},R_{1,2},\cdots,R_{1,n}) \in \mathbb{R}^{n\times n}.
\end{align}
The model introduced above will be used in the analysis and designs in the subsequent sections.

\section{Motivation}
\label{s:mot}

\begin{figure}[t]
	\centering
	\includegraphics[trim = 3.5mm 2mm 8.5mm 10mm, clip, width=\linewidth]{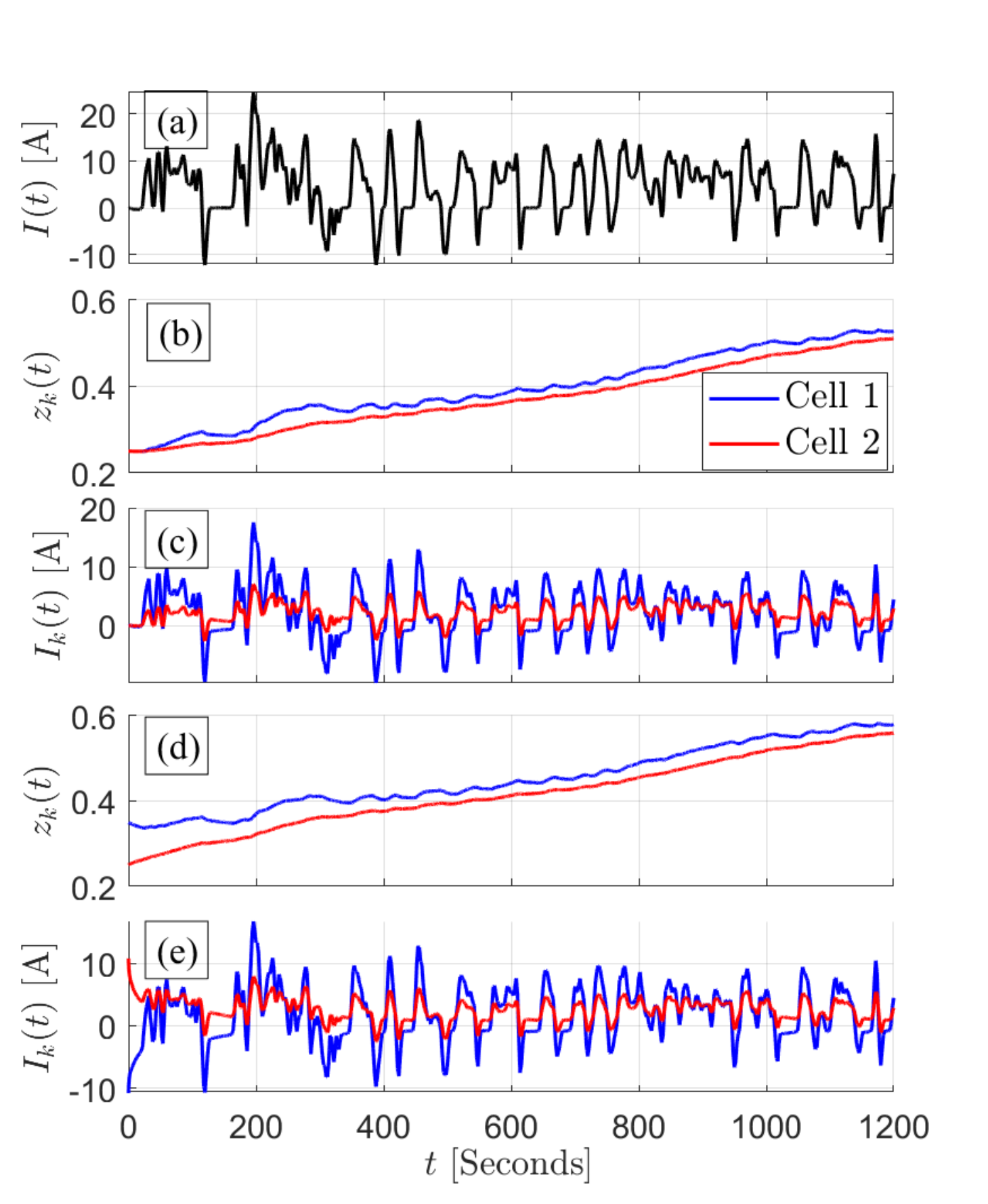}
	\vspace{-0.8cm}
    \caption{Simulation results of 2 cells in parallel with dissimilar model parameters and initial conditions. The total current is distributed unevenly due to cell heterogeneity.}
	%\vspace{-4ex}
	\label{Motivation}
\end{figure}

In this section, we demonstrate heterogeneity for cells connected in parallel via an open-loop simulation. We study a block configuration of two LiNiMnCoO$_2$/Graphite (NMC) type cells with 2.8 Ah nominal capacity in parallel. The heterogeneity arises from differences in model parameters, namely $C_k, R_{1,k}$ and $R_{2,k}$ for $k \in \{1,2\}$, and difference in SOC initialization. A transient electric vehicle-like charge/discharge cycle generated from urban dynamometer driving schedule (UDDS) is applied. Specifically, the total applied current (summation of local currents) is plotted in Fig. \ref{Motivation}(a).

Two cases are examined here. In the first case, the cells are initialized at the same SOC, $z_k(0) = 0.25$ for $k \in \{1,2\}$, but they differ in model parameters. Since Cell 2 has larger resistance, its local current is smaller in magnitude relative to local current of Cell 1, while the summation of the local currents equals to the total applied current for all $t$, as shown in Fig. \ref{Motivation}(b) and (c). In the second case, illustrated in Fig. \ref{Motivation}(d) and (e), the cells are initially different in both SOC initialization ($z_1(0)$ at 0.35 and $z_2(0)$ at 0.25) and model parameters. It can be observed that even though the applied total current is small (around zero) initially, Cell 1 takes large negative current (around -10 A) and Cell 2 positions itself at a large positive current (around 10 A). This occurs because $z_1(0)$ is initialized higher, and even though the $z$ values for two cells follow a similar trend, they do not synchronize. 
Since conventional BMSs do not monitor the local current of each parallel cell, some cells might be operating outside their safe operating region. Therefore, it will be of significant value to estimate and monitor the local currents and SOCs caused by cell heterogeneity to ensure safe battery pack operation.

%%%%%%%%%%%%%%%%%%%%%%%%%%%%%%%%%%%%%%%%%%%%%
\section{Observability Analysis} \label{s:obs_ana}

In this section, we mathematically analyze the observability of the nonlinear descriptor system with the input-output setup, via (i) linearization and (ii) Lie algebra.

\subsection{Observability from Linearization}

In order to study the observability of the nonlinear descriptor system \eqref{descriptor_dyn}-\eqref{descriptor_out}, 
we linearize the system around an equilibrium point $w = \overline{w}$ and check the observability conditions for the linearized system. 
If the linearized system is observable at $w = \overline{w}$, then the nonlinear system is locally observable. However, the observability conditions arising from linearizing the nonlinear system can be conservative, and
nothing can be concluded for the nonlinear system if the linearized system is \emph{not} observable. The linearized model of \eqref{descriptor_dyn}-\eqref{descriptor_out} takes the form
\begin{align}
    E \dot{w}(t) = & F w(t) + B u(t), \label{e:linstate} \\
    y(t) = & C w(t), \label{e:linout}
\end{align}
where the state matrix $F\in\mathbb{R}^{3n\times3n}$ and output matrix $C\in\mathbb{R}^{n\times3n}$ are given by
\begin{equation}
F = \left. A + \frac{d \theta}{d w}(w)\right\vert_{w=\overline{w}}\!\!\!\!, \quad
C = \left. H + \frac{d \phi}{d w}(w)\right\vert_{w=\overline{w}}.
\end{equation}
with matrix $A$ and $H$ provided in \eqref{e:Amatrix} and \eqref{e:Hmatrix}, respectively.

\iffalse
Therefore, matrix $F$ takes the form
\begin{equation}
    F = \begin{bmatrix} A_{11} & A_{12} \\ F_{21} & A_{22} \end{bmatrix} \in \mathbb{R}^{3n \times 3n}, \label{e:Fmatrix}
\end{equation}
with $F_{21}\in \mathbb{R}^{n \times 2n}$ given by
\begin{equation}
F_{21} = \begin{bmatrix}
    \displaystyle  \frac{dOCV}{d z_1} & 1 & \displaystyle -\frac{dOCV}{d z_2} & -1 & 0 &\cdots & 0 \vspace{0.1cm}\\
    \displaystyle  \frac{dOCV}{d z_1} & 1 & 0 & \displaystyle -\frac{dOCV}{d z_3} & -1 &\cdots & 0 \\
    \vdots & \vdots & \vdots & \vdots & \vdots & \ddots & \vdots \\
    \displaystyle \frac{dOCV}{d z_1} & 1 & 0 & 0 & \cdots & \displaystyle -\frac{dOCV}{d z_n} &-1 \\
    0 & 0 & 0 & 0 & \cdots & 0 & 0 \\
    \end{bmatrix} ,
\end{equation}
and matrix $C\in\mathbb{R}^{n\times 3n}$ takes the form
\begin{equation}
    C = \begin{bmatrix}
    C_x & H_u
    \end{bmatrix}
    \label{e:Cmatrix}
\end{equation}
with $C_x \in \mathbb{R}^{n\times 2n}$ given by 
% \begin{equation}
%     C_x = \begin{bmatrix}
%     \displaystyle \frac{d OCV}{d z_1} & 1 & 0 & \cdots & 0
%     \end{bmatrix}
% \end{equation}
\begin{equation}
C_x = \begin{bmatrix}
    \displaystyle  \frac{dOCV}{d z_1} & 1 & 0 & 0 &\cdots & 0 \vspace{0.1cm}\\
    0 & 0 & \displaystyle \frac{dOCV}{d z_2} & 1 &\cdots & 0 \\
    \vdots & \vdots & \vdots & \vdots & \ddots & \vdots \\
    0 & 0 & 0 & \cdots & \displaystyle \frac{dOCV}{d z_n} & 1 \\
    \end{bmatrix}.
\end{equation}
\fi

Let us now introduce the definition of complete observability (C-observability) for the descriptor system \eqref{e:linstate}-\eqref{e:linout}.

\begin{theorem} [\cite{duan2010analysis}]
\label{thmcobs}
The regular descriptor linear system \eqref{e:linstate}-\eqref{e:linout} is C-observable if and only if the following two conditions hold:
\begin{enumerate}
    \item[C.1] $\mathrm{rank}\left\{ [ E^\top, C^\top ]^\top \right\} = 3n$;
    \item[C.2] $\mathrm{rank}\left\{ [ (s E - F)^\top, C^\top ]^\top \right\} = 3n, \ \forall s \in \mathbb{C}$.
\end{enumerate}
\end{theorem}

Condition C.1 concerns C-observability of the fast (algebraic) subsystem while C.2 involves the slow (dynamic) subsystem.
Focusing first on condition C.1, it can be verified by construction since all cells in a parallel connection have the same voltage and therefore it is equivalent to measure the voltage of cell $i$ or $j$, with $i \neq j$.
Looking at condition C.2, it can be verified if the considered battery cells in parallel are different in terms of any model parameter among $\{ Q, R_1, R_2, C, {\rm d}OCV/{\rm d}z \}$. 
Finally, ${\rm d}OCV/{\rm d}z \neq 0$ is also required in order to guarantee condition C.2, i.e. if one of the cell's OCV curves becomes flat, then the observability of the dynamic linearized subsystem is lost. 
%These conditions are only sufficient for \bi{(local)} C-observability of the nonlinear descriptor system \eqref{descriptor_dyn}-\eqref{descriptor_out}. 
Notice that the verification of condition C.2 requires the numerical computation of the generalized eigenvalues of the pair $(E,A)$.

\subsection{Local Observability from Lie Algebra}

To elucidate if less conservative observability conditions exist, we analyze local observability for the nonlinear system resulting from a reduced descriptor system. That is, we analyze the system that results from eliminating the algebraic states through substitution. System \eqref{descriptor_dyn}-\eqref{descriptor_out} can be divided into differential and algebraic states with an explicit input current, i.e.
\begin{align}
    \label{descriptor_dynb}
    \hspace{-0.2cm}
    \left[\!\!\!\!\begin{array}{cc}
    I_{2n\times2n} &\!\! 0 \\
    0 &\!\! 0
    \end{array}\!\!\!\!\right] \dot{w}(t) 
    \!\!= \!\!
    &\left[\!\!\!\!\begin{array}{cc}
    A_{11} &\!\! A_{12} \\
    A_{21} &\!\! A_{22}
    \end{array}\!\!\!\!\right] w(t) \!+\!\!
    \left[\!\!\!\!\begin{array}{c}
    0 \\
    \theta_I
    \end{array}\!\!\!\!\right] I(t) \!+\!\!     \left[\!\!\!\!\begin{array}{c}
    \theta_x(w) \\
    \theta_{OCV}(w)
    \end{array}\!\!\!\!\right]\!\!, \\
    \label{descriptor_outb}
    y(t) = & [H_1 \;\; H_2]w(t) + \phi(w).
\end{align}
where $\theta_I = [0, \cdots, -1]^\top \in \mathbb{R}^n$ and $\theta_{OCV}(w) = \theta_u(w) - \theta_I I(t)$. 
Eq. \eqref{descriptor_dynb} can be solved for the algebraic state, resulting in the following transformation
\begin{equation}
    \label{e:ustate}
    u = - A_{22}^{-1}\left(A_{21} x(t) + \theta_I I(t) + \theta_{OCV}(w) \right)
\end{equation}
where matrix $A_{22}$ in \eqref{e:A22mat} is non-singular. Notice that \eqref{e:ustate} is an explicit solution for the algebraic state $u(t)$, i.e. the $w$-dependent functions $\theta_{OCV}(w)$ and $\phi(w)$ are now only dependent on the dynamic state $x(t)$. With an abuse of notation, these functions are denoted as $\theta_{OCV}(x)$ and $\phi(x)$ in the remainder of this section. For its part, $\theta_x(w) = \mathbf{0}_{2n \times 1} = \theta_x$.

Substituting \eqref{e:ustate} back into the differential part of the state equation \eqref{descriptor_dynb} and the output equation \eqref{descriptor_outb} yields the following nonlinear (control affine) reduced model
\begin{align}
    \label{descriptor_dynd}
    \dot{x}(t) = & f(x(t)) + g(x(t)) I(t), \\
    \label{descriptor_outd}
    \overline{y}(t) = & h(x(t)),
\end{align}
with
\begin{align}
    \label{ffunc}
    f(x) = & (A_{11} - A_{12} A_{22}^{-1} A_{21}) x(t) - A_{12} A_{22}^{-1} \theta_{OCV}(x) + \theta_x, \\
    g(x) = & - A_{12} A_{22}^{-1} \theta_I, \\
    \label{descriptor_outc}
    h(x) = & (H_{1} - H_{2} A_{22}^{-1} A_{21}) x(t) - H_{2} A_{22}^{-1} \theta_{OCV}(x) + \phi(x),
\end{align}
% \begin{align}
%     \label{descriptor_dync1}
%     \dot{x}(t) = & A_x x(t) +
%     B_{x1} I(t) +     
%     B_{x2} \theta_x(w(t)), \\
%     \label{descriptor_outc1}
%     y(t) = & H_x x(t) + D_{x1} I(t) + D_{x2} \theta_x(w(t)) + \phi(w(t)).
% \end{align}
where the output $\overline{y}(t) = y(t) - H_{2} A_{22}^{-1} \theta_I I(t)$.

Let us now introduce the notion of local observability \cite{Vidy-2002} for system \eqref{descriptor_dynd}-\eqref{descriptor_outd}. 
\begin{theorem} \label{thm2}
The system \eqref{descriptor_dynd}-\eqref{descriptor_outd} is locally observable around $x_0 \in X$ if there exists $n$ linearly independent row vectors in the set
\begin{equation}
    ({\rm d}L_{z_s}L_{z_{s-1}} \ldots L_{z_1} h_j)(x_0) \nonumber
\end{equation}
where $L_{z_s}h_j$ are Lie derivatives of $h_j$ with respect to $z_s$ and ${\rm d} h_j$ is the gradient of $h_j$ to be defined below, $s\geq0$ and $z_i\in\{f,g\}$, with $j = 1,\ldots,p$ ($p$ is the number of outputs) and $i = 1,\ldots,s$ (for $s=0$, the set comprises ${\rm d}h_j(x_0)$).
\end{theorem}

In Theorem \ref{thm2}, the gradient of $h_j$ and Lie derivatives of $h_j$ with respect to function $f$ are given by
\begin{align}
    {\rm d}h_j &= \left[ \frac{\partial h_j}{\partial x_1} \;\;
    \frac{\partial h_j}{\partial x_2} \;\; \cdots \;\;
    \frac{\partial h_j}{\partial x_n} \right], \nonumber \\
    L_{f}h_j & = \langle {\rm d}h_j,f \rangle,
\end{align}
and similarly for function $g$. 
The $0$th order Lie derivative $L_f^0 h_j$ is defined as $h_j$ whereas the $2$nd order Lie derivative takes the form $L_f^2 h_j = L_f L_f h_j$.

According to Theorem \ref{thm2}, the observability rank condition can be derived. Define the observability matrix as
\begin{equation}
    \mathcal{O} = 
    \left[
    {\rm d}h(x) \;\;\;
    {\rm d}L_fh(x) \;\;\;
    {\rm d}L_gh(x) \;\;\;
    {\rm d}L_f^2h(x) \;\;\;
    {\rm d}L_g^2h(x) \;\;\;
    \cdots
    \right]^\top.
\end{equation}
Then the model \eqref{descriptor_dynd}-\eqref{descriptor_outd} is locally observable around $x_0$ if ${\rm rank}(\mathcal{O}) = 2n.$ Note first that the observability matrix $\mathcal{O}$ is not bounded from above. Secondly, this matrix depends on states $x_k = [z_k, V_{c,k}]^\top$, functions $OCV_k(z_k)$ as well as parameters $\theta_k = [R_{1,k}, \tau_k, C_k, Q_k]^\top$, where $\tau_k = -1/(R_{2,k}C_k)$.

To keep the observability analysis tractable, we consider two cells in parallel, i.e. $k\in\{1,2\}$. 
From the observability rank condition, we conclude: (i) the system is locally observable at $x_0$ if cells have different parameter values $\theta_k$, and (ii) observability conditions are not fulfilled if the cells are completely equivalent, i.e. all parameter values are uniform across the cells. In the second case, the cells are presumably indistinguishable and a single cell model can be utilized to represent the parallel connection. In between these two extremes, the observability conditions cannot be verified if any of the following conditions hold:
\begin{enumerate}
    \item the parameters $\tau_1 = \tau_2$ AND $R_{1,1} Q_1 = R_{1,2} Q_2$ AND $R_{1,1} C_1 = R_{1,2} C_2$;
    \item the functions $OCV(z_1) = OCV(z_2)$ AND $dOCV(z_1)/dz_1 = dOCV(z_2)/dz_2$;
    \item at least one of the $l$-th OCV derivatives 
    %all the OCV derivatives of one of the cells are zero simultaneously.
    satisfy $d^lOCV(z_j)/dz_j^l = 0$, for the $j$-th cell and $l=1,\ldots,\infty$.
\end{enumerate}

Note that a classical approach to study observability of a single ECM is to linearize the model, as done in e.g. \cite{Rausch-2013}, for the case of two battery cells in parallel. By doing so, the ECM is observable if ${\rm d}OCV/{\rm d}z\neq 0$. This condition on the first OCV derivative is conservative as it was found in \cite{Zhao-2017} through the local observability analysis of a single cell. 
This more detailed analysis showed that an OCV derivative must be different than zero to guarantee local observability, but it does not need to be the first derivative. 
This fact was also verified above with a similar analysis for two ECMs in parallel.

A similar observability analysis as the one proposed here was also carried out in \cite{Lin-2015} considering cells in series. 
However, the observability matrix for a series string \mbox{differs} from that of a parallel arrangement. Namely, in the series arrangement, each cell's parameters/states appear in a column of $\mathcal{O}$. 
This is not the case for a parallel topology, where parameters/states of the cells are scattered all over the different entries in $\mathcal{O}$. %, as can be seen in \eqref{e:dh}-\eqref{e:dLf2h}. 
Therefore, parameters/states of one cell influence the local observability of the neighbouring cell in a parallel arrangement.

When compared to the observability analysis of \mbox{Theorem \ref{thmcobs}} based on the linearized descriptor system, the local observability analysis of the nonlinear system using \mbox{Theorem \ref{thm2}} is less conservative \cite{Vidy-2002} and more informative. The latter aspect relies on the fact that the observability matrix $\mathcal{O}$ explicitly depends on model parameters, and it can be analytically obtained through e.g. symbolic software.
%can be computed, and the impact of different parameter combinations can be assessed. 
%On the other hand, less restrictive conditions, like the one related to the OCV function and its derivatives introduced above, can be obtained.

%%%%%%%%%%%%%%%%%%%%%%%%%%%%%%%%%%%%%%%%%%%%%%
\section{Design of State Observers} \label{s:obs_design}

The following observer with linear output error injection is proposed for the plant model (\ref{descriptor_dyn})-(\ref{descriptor_out})
\begin{align}
    \label{obs_dyn}
    E\dot{\hat{w}} & = A\hat{w} + \theta(\hat{w}) + K(y-H\hat{w}-\phi(\hat{w})), \\
    \label{obs_out}
    \hat{y} & = H\hat{w} + \phi(\hat{w}),
\end{align}
where $K \in \mathbb{R}^{3n}$ is the observer gain vector to be designed and $\hat{w}$ is the estimation for $w$. The following theorem based on \cite{kaprielian1992observer,rajamani1998observers} establishes the convergence results of the proposed observer.

\begin{theorem} \label{thm_obs}
Consider the plant model dynamics (\ref{descriptor_dyn})-(\ref{descriptor_out}), and suppose the matrix $\begin{bmatrix}A_{22} & H_u\end{bmatrix}^\top$ has rank $n$. Let
\begin{equation}
    G = (A-KH) = \begin{bmatrix} G_{11} & G_{12} \\ G_{21} & G_{22} \end{bmatrix},
\end{equation}
and define the matrix
\begin{equation}
    \widetilde{G} = (G_{11}-G_{12}G_{22}^{-1}G_{21}).
\end{equation}
Suppose the function
\begin{equation}
    f(w) = \theta_x(w)-G_{12}G_{22}^{-1}\theta_u(w) + (G_{12}G_{22}^{-1}K_u-K_x)\phi(w),
\end{equation}
is Lipschitz continuous with respect to $x$, i.e.,
\begin{align}
    \|f(x_1,u) - f(x_2,u)\| & \leq \gamma \|x_1 - x_2\|,
\end{align}
where $\gamma \in \mathbb{R}$ is the Lipschitz constant. If the observer gain $K$ is chosen to ensure that $\widetilde{G}$ is stable, and
\begin{equation}
    \min_{\omega \in \mathbb{R}^+} \lambda_{\rm min}(A-KH-j\omega I_{3n \times 3n}) > \gamma,
\end{equation}
then the zero equilibrium of the dynamics of estimation error  $e(t) = w(t) - \hat{w}(t)$ given by
\begin{equation}
    \label{error_dyn}
    E\dot{e} = Ge + \theta(w) - \theta(\hat{w}) - K[\phi(w) - \phi(\hat{w})]
\end{equation}
is asymptotically convergent to zero.
\end{theorem}

\begin{proof}
Let the state estimation error $e = \begin{bmatrix} e_x & e_u \end{bmatrix}^\top$, with $e_x = x - \hat{x}$ being the estimation error for the differential states and $e_u = u - \hat{u}$ the estimation error for the algebraic states. Then (\ref{error_dyn}) can be written as
\begin{align}
    \label{error_sys}
    \begin{bmatrix} I_{2n \times 2n} & \mathbf{0} \\ \mathbf{0} & \mathbf{0} \end{bmatrix} \begin{bmatrix} \dot{e}_x \\ \dot{e}_u  \end{bmatrix} = & \begin{bmatrix} G_{11} & G_{12} \\ G_{21} & G_{22} \end{bmatrix} \begin{bmatrix} e_x \\ e_u \end{bmatrix} + \begin{bmatrix} \theta_x(w) - \theta_x(\hat{w}) \\ \theta_u(w) - \theta_u(\hat{w})  \end{bmatrix} \nonumber \\
    & - \begin{bmatrix} K_x \\ K_u \end{bmatrix}[\phi(w) - \phi(\hat{w})].
\end{align}
We highlight that $G_{22}$ can be non-singular (i.e., invertible) if the linear part of (\ref{descriptor_dyn}) is impulse observable \cite{kaprielian1992observer}, i.e., the matrix $\begin{bmatrix}A_{22} & H_u\end{bmatrix}^\top$ has rank $n$. Then the estimation error system (\ref{error_sys}) is equivalently described by
\begin{align}
    \label{ex_dyn}
    \dot{e}_x = & (G_{11}-G_{12}G_{22}^{-1}G_{21})e_x + [\theta_x(w)-G_{12}G_{22}^{-1}\theta_u(w)] \nonumber \\
    & - [\theta_x(\hat{w})-G_{12}G_{22}^{-1}\theta_u(\hat{w})] + (G_{12}G_{22}^{-1}K_u-K_x)\phi(w) \nonumber \\
    & - (G_{12}G_{22}^{-1}K_u-K_x)\phi(\hat{w}) \nonumber \\
    = & \widetilde{G}e_x + f(w) - f(\hat{w}),
\end{align}
along with the algebraic equation
\begin{align}
    e_u = & -G_{22}^{-1}G_{21}e_x - G_{22}^{-1}[\theta_u(w) - \theta_u(\hat{w})] \nonumber \\
    & + G_{22}^{-1}K_u[\phi(w) - \phi(\hat{w})].
\end{align}

Consider the following Lyapunov function for the error system (\ref{ex_dyn}), corresponding to the differential states $e_x$,
\begin{equation}
    \label{Lyap}
    W(t) = \frac{1}{2}e_x^\top P e_x.
\end{equation}
The derivative of the Lyapunov function $W(t)$ along the trajectory of $e_x$ is computed by
\begin{align}
    \dot{W} = & \frac{1}{2}\dot{e}_x^\top P e_x + \frac{1}{2}e_x^\top P \dot{e}_x \nonumber \\
    = & \frac{1}{2}e_x^\top(\widetilde{G}^\top P+P\widetilde{G})e_x + e_x^\top P [f(w) - f(\hat{w})] \nonumber \\
    \leq & \frac{1}{2}e_x^\top (\widetilde{G}^\top P+P\widetilde{G}) e_x + \|Pe_x\|[\|f(w) - f(\hat{w})\|] \nonumber \\
    \leq & \frac{1}{2}e_x^\top (\widetilde{G}^\top P+P\widetilde{G}) e_x + \gamma\|Pe_x\|\|e_x\| \nonumber \\
    \leq & \frac{1}{2}e_x^\top[\widetilde{G}^\top P+P\widetilde{G}+\gamma^2PP+I]e_x,
\end{align}
where the inequality
\begin{equation}
    2\gamma\|Pe_x\|\|e_x\| \leq \gamma^2e_x^\top PP e_x + e_x^\top e_x
\end{equation}
has been utilized. According to Theorem 2 in \cite{rajamani1998observers}, the estimation error $e_x$ is asymptotically stable if the conditions of Theorem \ref{thm_obs} hold. Under this scenario, when $t \rightarrow \infty$, [$\theta_u(w) - \theta_u(\hat{w})] \rightarrow \mathbf{0}_{n \times 1}$, and $[\phi(w) - \phi(\hat{w})] \rightarrow 0$. Hence, the estimation error $e_u$ for the algebraic states also converge to zero asymptotically.
\end{proof}

%%%%%%%%%%%%%%%%%%%%%%%%%%%%%%%%%%%%%%%%%%%%%%
\section{Simulation Study} \label{s:sim}
%%%%%%%%%%%%%%%%%%%%%%%%%%%%%%%%%%%%%%%%%%%%%%

\begin{table}[t]
    \caption{Model Parameters in Simulation Study} \label{params}
	\centering
	\begin{tabular}{  c c c c }
	\hline \hline
	 & Cell 1 & Cell 2 & Units \\ \hline
	$R_{1,k}$ & 0.0025 & 0.0015 & [$\Omega$] \\
    $R_{2,k}$ & 0.004 & 0.0035 & [$\Omega$]  \\
    $C_k$ & 1500 & 2000 & [F] \\
    $Q_k$ & 2.3 & 2.0 & [Ah] \\
    $z_0$ & 0.4 & 0.5 & [--]  \\
    \hline \hline
	\end{tabular}
\end{table}

A simulation study using a battery block with $n = 2$ NMC cells connected in parallel is conducted to evaluate the performance of the proposed estimation scheme. Without loss of generality, a pair of cells is preferred over a larger block to facilitate the presentation of results. 
We consider the situation in which the cells may differ in their initial SOCs and model parameters, but subject to the same SOC-OCV relationship. The considered setup guarantees local observability based on the analysis in Section \ref{s:obs_ana}-A.

The model parameters and initial SOCs are shown in Table \ref{params}. Under these circumstances, the state vector is given by
\begin{equation}
    w = \begin{bmatrix}z_1 & V_{c,1} & z_2 & V_{c,2} & I_1 & I_2\end{bmatrix}^\top \in \mathbb{R}^6.
\end{equation}
In this simulation study, the total applied current comes from a UDDS drive cycle provided in Fig. \ref{Motivation}(a) with appropriate scaling.
The observer in (\ref{obs_dyn})-(\ref{obs_out}) is used to estimate the individual cell SOCs and the local currents by using only the voltage and overall current measurements. The initial SOC errors between the plant model and the observer are 15\% and 10\%, respectively. The observer gain is chosen to be $L = \left[-30\;-30\;-20\;\;2\;\;4\;-20\right]^\top$,
which satisfies the conditions of Theorem \ref{thm_obs}.

\begin{figure}[t]
	\centering
	\includegraphics[trim = 6mm 0mm 9.5mm 5mm, clip, width=\linewidth]{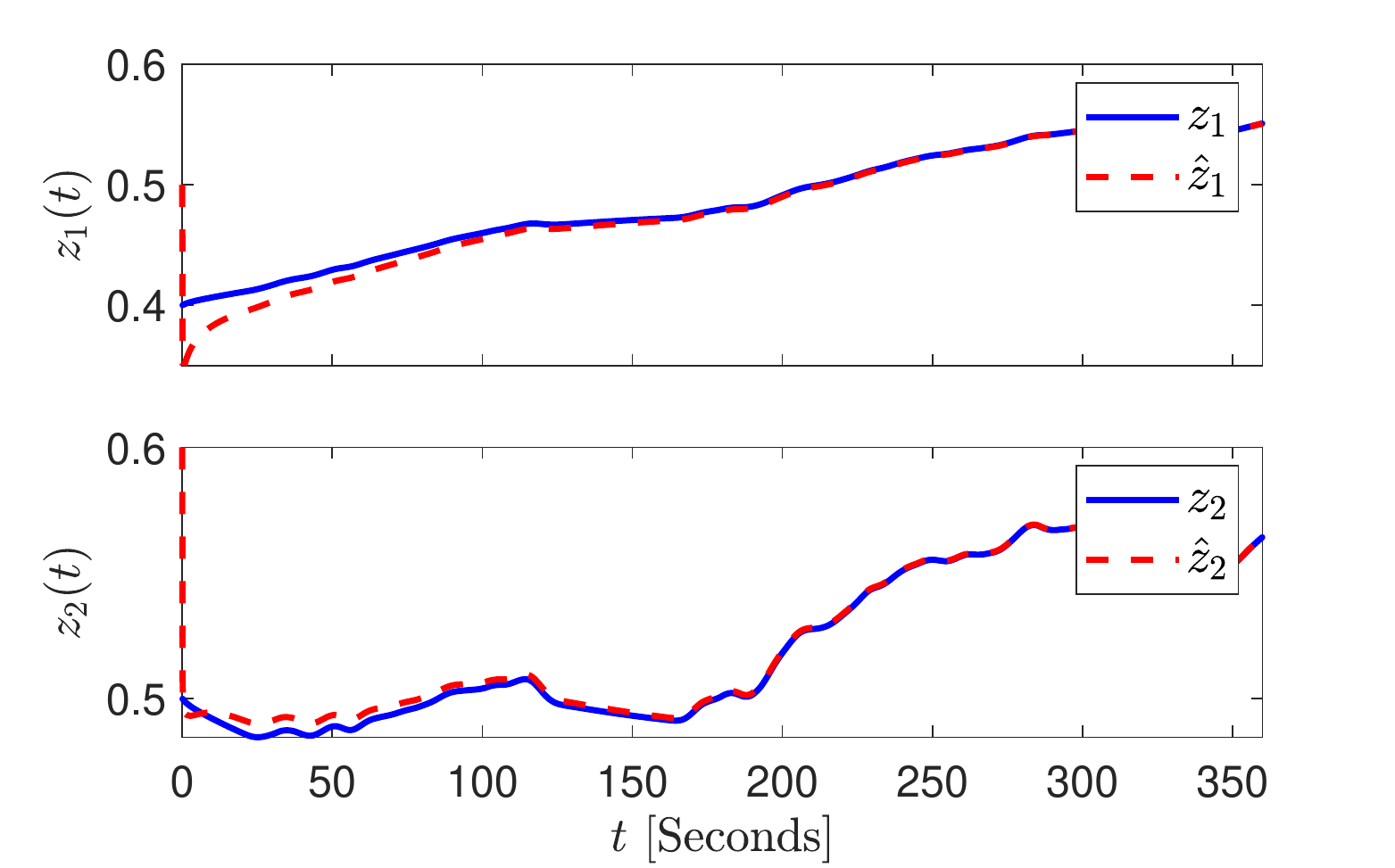}
	\caption{The estimation performance for SOCs of a two-cell parallel configuration. The results verify the asymptotic convergence.}
	%\vspace{-4ex}
	\label{fig:SOC}
\end{figure}

\begin{figure}[t]
	\centering
	\includegraphics[trim = 7mm 0mm 9.5mm 5mm, clip, width=\linewidth]{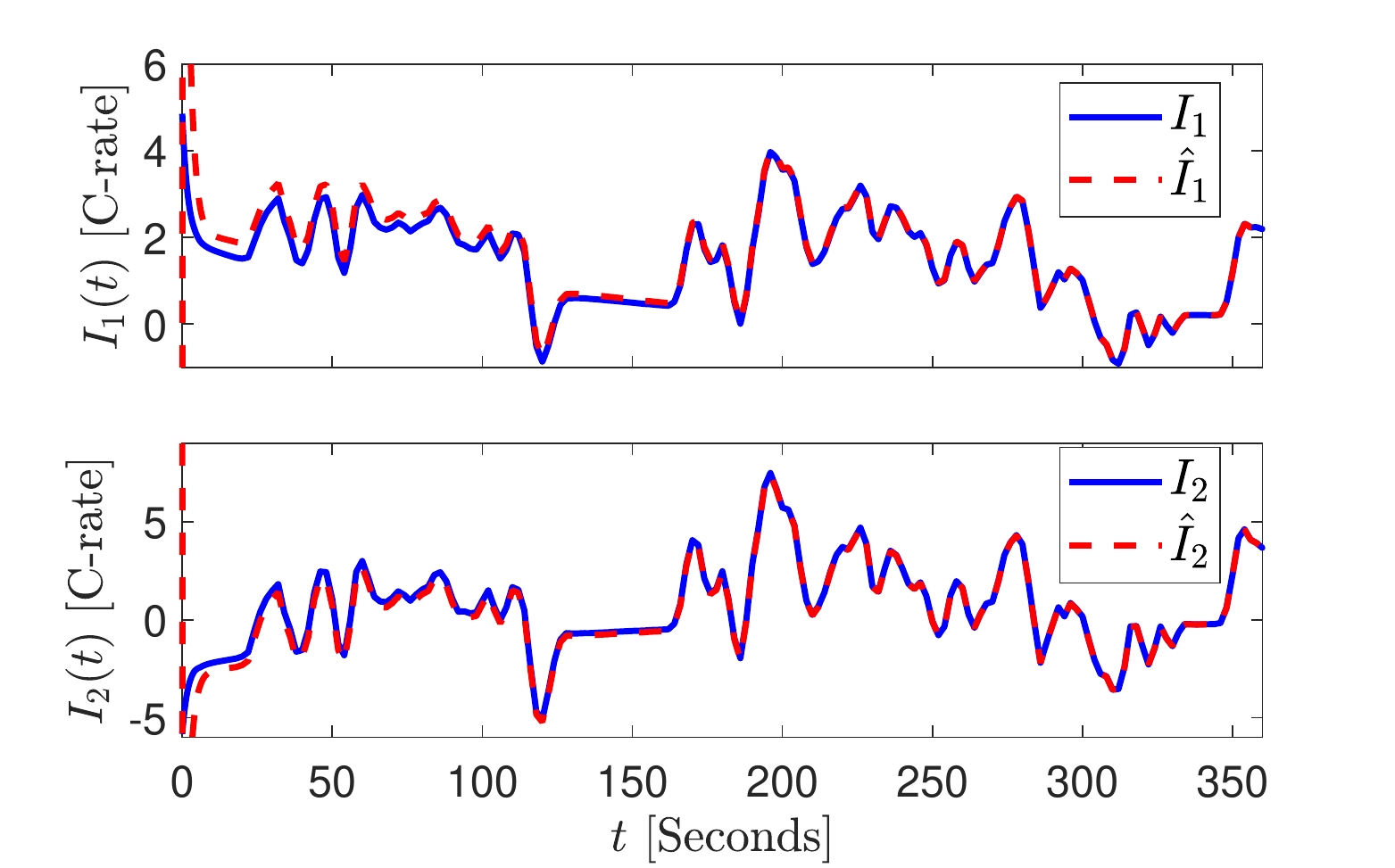}
	\caption{The estimation performance for local currents of a two-cell parallel configuration. The results verify the asymptotic convergence.}
	%\vspace{-4ex}
	\label{fig:Current}
\end{figure}

Figs. \ref{fig:SOC} and \ref{fig:Current} demonstrate estimation performance, where the solid blue curves are the true states from the plant model and the dashed red curves are the estimated ones. 
Fig. \ref{fig:SOC} displays the estimates for the differential states (SOCs), whereas Fig. \ref{fig:Current} portrays the estimates for the algebraic states (local currents). The state observer is able to recover the true signals quickly (in approximately 100 seconds) from large initial estimation errors. These results confirm the asymptotic zero error convergence conclusions in Theorem \ref{thm_obs}.

%%%%%%%%%%%%%%%%%%%%%%%%%%%%%%%%%%%%%%%%%%%%%%
\section{Conclusion} \label{s:conclusion}
%%%%%%%%%%%%%%%%%%%%%%%%%%%%%%%%%%%%%%%%%%%%%%

A nonlinear descriptor system has been proposed to model parallel arrangements of lithium-ion battery cells, and a state observer for such systems has been developed. 
This modelling framework fits naturally with battery applications, given the interconnections arising from Kirchhoff's laws. 
The design procedure used to build the state observer from this model avoids linearization or canonical transformations, and it only relies on the assumption of Lipschitz nonlinearities. 
The resulting state observer benefits from considering the cell currents as algebraic states to be simultaneously estimated with the differential states. 
The effectiveness of the proposed estimation approach was demonstrated in simulation.

\bibliographystyle{ieeetr}
\bibliography{Reference} 

\end{document}